\documentclass{jgm}
\usepackage{latexsym,amssymb,amsmath,amsthm,amsfonts,enumerate,verbatim,xspace,exscale}
\usepackage{amsbsy}
\usepackage[colorlinks=true]{hyperref}
 \hypersetup{urlcolor=blue, citecolor=red}

\def\currentvolume{1}
 \def\currentissue{1}
  \def\currentyear{2009}
   \def\currentmonth{March}
    \def\ppages{1--XX}
     \def\DOI{10.3934/jgm.2009.1.35}

\theoremstyle{plain}
\newtheorem{theorem}{Theorem}[section]

\newtheorem{proposition}[theorem]{Proposition}
\newtheorem{corollary}[theorem]{Corollary}

\theoremstyle{definition}

\renewcommand{\thesubsection}{\thesection.\Alph{subsection}}
\renewcommand{\theequation}{\thesection.\arabic{equation}}

\newcommand{\g}{\mathfrak{g}}

\def\O{\Omega}

\def\so{\mathfrak{so}}

\def\g{\mathfrak{g}}

\def\D{\mathcal{D}}
\def\R{\mathbb{R}}

\def\L{\mbox{Leg}}
\def\M{\mathcal{M}}

\def\A{\mathcal{A}}
\def\Ham{\mathcal{H}}

\def\C{\mathcal{C}}

\newcommand{\caps}[1]{\textup{\textsc{#1}}}

\providecommand{\bysame}{\makebox[3em]{\hrulefill}\thinspace}

\newcommand{\up}{\upshape}

\newcommand{\longto}{\longrightarrow}
\newcommand{\hookto}{\hookrightarrow}
\newcommand{\toto}{\twoheadrightarrow}

\def\vv<#1>{\langle#1\rangle}

\newcommand{\tr}{\mbox{$\textup{Tr}$}}

\newcommand{\id}{\mbox{$\text{\up{id}}\,$}}

\providecommand{\det}{\mbox{$\text{\up{det}}\,$}}

\newcommand{\dd}[2]{\mbox{$\frac{\partial #2}{\partial #1}$}}

\newcommand{\om}{\omega}
\newcommand{\Om}{\Omega}
\newcommand{\var}{\varphi}

\newcommand{\lam}{\lambda}
\newcommand{\Lam}{\Lambda}

\newcommand{\wt}[1]{\mbox{$\widetilde{#1}$}}

\newcommand{\by}[2]{\mbox{$\frac{#1}{#2}$}}

\providecommand{\set}[1]{\mbox{$\{#1\}$}}

\newcommand{\X}{\mathfrak{X}}

\newcommand{\curv}{\mbox{$\textup{Curv}$}}

\newcommand{\ver}{\mbox{$\textup{Ver}$}}
\newcommand{\hor}{\mbox{$\textup{Hor}$}}

\newcommand{\momap}{momentum map\xspace}

\newcommand{\gu}{\mathfrak{g}}
\newcommand{\ho}{\mathfrak{h}}

\newcommand{\Ad}{\mbox{$\text{\upshape{Ad}}$}}

\newcommand{\orb}{\mbox{$\mathcal{O}$}}

\newcommand{\SO}{\mbox{$\textup{SO}$}}

\newcommand{\hlphi}{\mbox{$\textup{hl}^{\mathcal{A}}$}}
\newcommand{\hlphit}{\mbox{$\textup{hl}^{\widetilde{\mathcal{A}}}$}}
\newcommand{\Hamc}{\mbox{$\mathcal{H}_{\textup{c}}$}}
\newcommand{\Omnh}{\mbox{$\Om_{\textup{nh}}$}}
\newcommand{\Xnh}{\mbox{$X_{\textup{nh}}$}}

\title[$G$-Chaplygin systems, truncation and Chaplygin's ball]{$G$-Chaplygin systems
with internal symmetries,
truncation, and an (almost) symplectic view of Chaplygin's ball}

\author[Simon Hochgerner and Luis Garc{\'i}a-Naranjo]{}
\email{simon.hochgerner@epfl.ch}
\email{luis.garcianaranjo@epfl.ch}

\keywords{Chaplygin's ball, non-holonomic systems, Hamiltonization}
\subjclass{Primary: 70F25; Secondary: 53D20}

\begin{document}

\maketitle

\centerline{\scshape Simon Hochgerner and Luis Garc{\'i}a-Naranjo}
\medskip
{\footnotesize
 \centerline{Section de Mathematiques}
   \centerline{Station 8, EPFL}
   \centerline{CH-1015 Lausanne, Switzerland}
} 

%

\bigskip

 \centerline{(Communicated by Jair Koiller)}

\begin{abstract}
Via compression (\cite{Koiller1992,EhlersKoiller})
we write the $n$-dimensional Chaplygin sphere system as an
almost Hamiltonian system on $T^*\SO(n)$ with internal symmetry group
$\SO(n-1)$. We show how this symmetry group can be factored out, and
pass to the fully reduced system on (a fiber bundle over)
$T^*S^{n-1}$.
This approach yields
an explicit description of the reduced system in terms of the
geometric data involved. Due to this description we can study
Hamiltonizability of the system. It turns out that the homogeneous
Chaplygin ball, which is not Hamiltonian at the $T^*\SO(n)$-level,
is Hamiltonian at the $T^*S^{n-1}$-level. Moreover, the
$3$-dimensional
ball becomes Hamiltonian at the $T^*S^{2}$-level after
time reparametrization, whereby we re-prove a result of
\cite{BorisovMamaev,BorisovMamaev2005} in symplecto-geometric terms.
We also study compression followed by reduction of generalized
Chaplygin systems.
\end{abstract}

\tableofcontents

\section{Introduction and description of results}

A non-holonomic system (with linear constraints) consists of a configuration manifold
$Q$, a Lagrangian $L: TQ\to\R$, and a non-integrable
smooth distribution $\D\subset TQ$.
The equations of motion for a curve $q(t)$ in $Q$ are determined by the
Lagrange-d'Alembert principle (constraining force does not exert work)
supplemented by the condition that
$q'\in\D$.
We shall only deal with
constraint distributions $\D$ that are of constant rank.
Further, $L$ will be of the form `kinetic energy
minus potential' where the kinetic energy defines a
Riemannian metric $\mu$  on the configuration manifold.

A $G$-Chaplygin system is a non-holonomic system $(Q,\D,L)$ which is
invariant under a free and proper
action by a Lie group $G$ on $Q$ such that $\D$
defines a connection on the principal bundle $Q\toto Q/G$. It is not
required that $\D$ is the mechanical connection associated to
$\mu$. Under these assumptions the equations of motion can be written
in a particularly nice format. Let $S=Q/G$ be the reduced
configuration space, $\Om^S$ the canonical
symplectic form on $T^*S$, $J: T^*Q\to\gu^*$ the standard \momap  of the
lifted $G$-action on $T^*Q$, and $K\in\Om^2(Q,\gu)$ be the curvature
form associated to $\D$. Then the non-holonomic system can be reduced,
or \emph{compressed}, to a dynamical system on $T^*S$ with dynamics
given by the vector field $\Xnh$ which is defined by
\begin{equation}\label{e:comp-syst}
 i(\Xnh)\Omnh = d\Hamc
 \textup{ where }
 \Omnh := \Om^S-\vv<J,K>.
\end{equation}
Here $\Hamc: T^*S\to\R$ is the \emph{compressed Hamiltonian}; it is
the function induced from the Legendre transform of $L$.
The term $\vv<J,K>$ does make sense as a semi-basic two-form on $T^*S$
since ambiguities cancel out. The form $\Omnh$ is, in general, an almost symplectic
form, that is, it is non-degenerate and non-closed. We will thus view
the compressed system $(T^*S,\Omnh,\Hamc)$ as an \emph{almost
Hamiltonian} system.
See
\cite{Koiller1992,EhlersKoiller,BS93} and Section~\ref{S:ham-setting}.
(Our sign in \eqref{e:comp-syst} is different from that in \cite{EhlersKoiller}
because of our choice of sign in $\Om^S=-d\theta$: \cite{EhlersKoiller} choose
$\Om^S=d\theta$ whence for them $\Omnh = \Om^S+\vv<J,K>$.)

The present paper is only concerned with non-holonomic systems that
arise as $G$-Chaplygin systems, and the description of the dynamics
in terms of the above mentioned compression process will be our starting point.
The question arises whether the compressed system $(T^*S,\Omnh,\Hamc)$
is
Hamiltonizable: is there a positive function $f: S\to\R$ such that
$f\Omnh$ is closed? If this is the case one says that $\Omnh$ is
conformally symplectic. The interpretation is that one is looking for
an $s\in S$ dependent time reparametrization $d\tau = fdt$ so that the
system becomes Hamiltonian in the new time.
That is, the dynamics described by the vector field
$\by{1}{f}\Xnh(\wt{c}(\tau)) = \dd{\tau}{t}\dd{t}{c} =
\dd{\tau}{}\wt{c}(\tau)$, where $\wt{c}(\tau)=c(t)$, are Hamiltonian
in the usual sense with respect to $f\Omnh$.
Moreover, it follows that the volume form $f^{m-1}\Omnh^m$ ($m=\dim
S$) is preserved by
the flow of $\Xnh$. Conversely, when a preserved volume form $F\Omnh^m$
exists then $F^{\frac{1}{m-1}}$ is a candidate for a conformal
factor.
See the
discussion in Ehlers, Koiller, Montgomery and Rios~\cite{EhlersKoiller}.

The classical Chaplygin sphere problem (\cite{Chaplygin1987a})
is that of a dynamically balanced $3$-dimensional ball
that rolls on a horizontal table without slipping.
Dynamically balanced means that the geometric center
coincides with the center of mass. However, we do not suppose that the
mass distribution is homogeneous. The inertia matrix can be any symmetric
positive definite three by three matrix.
The no slip condition is
a non-holonomic constraint on the velocities. The ball is
allowed to rotate about its vertical axis.
The reduced equations were first found and integrated by
Chaplygin~\cite{Chaplygin1987a} in terms of hyper-elliptic
functions. A thorough study of the algebraic integrability is given in
Duistermaat~\cite{D04} where it is explicitly stated that the system is
not Hamiltonian.

Chaplygin's rolling ball is a $G$-Chaplygin system with configuration
space $Q=\SO(3)\times\R^2$, constraint distribution
$\D$, kinetic energy Lagrangian, and symmetry group $G=\R^2$. Thus
$\D$ defines a horizontal connection on $Q\toto S=\SO(3)$. See
Section~\ref{S:Chaplygin} for details.
In \cite{EhlersKoiller} the compression of this system to an almost
Hamiltonian system $(T^*S,\Omnh,\Hamc)$ is carried out.
Further,
\cite{EhlersKoiller} prove that this compressed system is not
Hamiltonizable (at the $T^*\SO(3)$-level), not even in the homogeneous
case. On the other hand,
Borisov and Mamaev~\cite{BorisovMamaev,BorisovMamaev2005} give
explicit formulas for a Poisson bracket
which allow to
write the (reduced) equations of motion for Chaplygin's ball as a true
Hamiltonian system.
(Their bracket is explained in geometric terms involving
\emph{affine almost Poisson structures} in \cite{N08}.)
Their result is all the more remarkable as it is in apparent
contradiction to the assertions of \cite{D04,EhlersKoiller}.

It actually seems to be a general phenomenon that integrable
non-holonomic systems are related to integrable Hamiltonian
systems. See also \cite{FJ04,Jov08,Jov09}. This  observation provides
an important motivation for a systematic study of Hamiltonization of
integrable non-holonomic systems.

Chaplygin's
rolling ball is the topic of Section~\ref{S:Chaplygin}.
We describe its compression $(T^*S,\Omnh,\Hamc)$ in detail,
write the system in the form
(\ref{e:comp-syst}),
and pay particular attention to
the fact that there remains a further symmetry group even after
compression. Indeed, rotation of the ball about its vertical axis
induces an $S^1$-action on the compressed phase space $T^*S$ that
preserves $\Hamc$ and $\Omnh$.
In emphasizing the role of the
$\vv<J,K>$-term and the
almost Hamiltonian point of view we follow very closely the
exposition of \cite{EhlersKoiller}.

The main theme of the present paper is to establish a synthesis between
the papers of \cite{BorisovMamaev,BorisovMamaev2005} and \cite{EhlersKoiller}.
The crucial idea (actually due to \cite{EhlersKoiller})
which is used in this note is that the compressed system should be
further reduced with respect to the induced $S^1$-action, and
Hamiltonization should be attempted afterwards on the ultimate reduced
space $T^*S^2 = T^*(S/S^1)$.
(This is in agreement with \cite{BorisovMamaev,BorisovMamaev2005}
since the symplectic leaves of their Poisson bracket can be realized
as magnetic cotangent bundles over $S^2$.)
The $S^1$-symmetries are generally referred to as internal symmetries of
the system. Describing the corresponding reduction procedure is
non-trivial and is the main
result of the paper. (See Theorems~\ref{thm:trunc} and \ref{prop:trunc}.)

This problem can
be stated also for higher dimensional Chaplygin balls. Let $S=\SO(n)$
be the shape space of the $n$-dimensional Chaplygin ball rolling on an
$n-1$-dimensional horizontal plane with internal symmetry group
$H=\SO(n-1)$. Internal symmetries are very well behaved in that they
give rise to conserved quantities: the standard \momap $J_H:
T^*S\to\ho^*$ with respect to the canonical form $\Om^S$ is constant
along flow lines of $\Xnh$. However, $J_H$ is \emph{not} the \momap
with respect to $\Omnh$. That is, for $\lam\in\ho^*$,
the restriction of $\Omnh$ to
$J_H^{-1}(\lam)$ does not define a basic two form on the bundle
$J_H^{-1}(\lam)\toto J_H^{-1}(\lam)/H_{\lam}$ whence the system does
not descend to the `would be' ultimate reduced space. This is true
already for $n=3$. Now the point of Theorem~\ref{prop:trunc} is that
$\Omnh$ can be \emph{truncated} in a way that does not affect the equations
of motion but does provide the correct \momap. Effectively we replace
$\Omnh$ by a new two form $\wt{\Om}$ that is non-degenerate,
$H$-invariant, and satisfies
\[
 i(\Xnh)\wt{\Om} = d\Hamc
 \textup{ as well as }
 i(\zeta_Y)\wt{\Om} = \vv<dJ_H,Y>
\]
for all $Y\in\ho$ where $\zeta_Y$ denotes the infinitesimal generator
associated  to $Y$.
Why the name truncation? To construct $\wt{\Om}$ we use an $H$-connection
on the principal bundle $T^*S\toto(T^*S)/H$ such that $\Xnh$ is horizontal.
We employ this connection to truncate the $\vv<J,K>$-term in such a way that
it becomes horizontal with respect to the $H$-action and we retain only the
necessary information about the dynamics.
In particular, Theorem~\ref{prop:trunc} gives an
explicit formula
\[
 \wt{\Om} = \Om^S - \vv<L,\curv^{\om}>
\]
where $\curv^{\om}\in\Om^2(S,\ho)$ is the curvature form associated to
the Hopf connection on $\SO(n)\toto \SO(n)/H = S^{n-1}$ and $L:
T^*S\to\ho^*$ is a certain mapping (related to angular velocity in the
space frame)
that coincides with $J_H$ if and
only if the ball is homogeneous.
Notice also that $\wt{\Om}$ is of the same format `canonical form minus
semi-basic' as $\Omnh$.
Now, one can carry out almost
Hamiltonian reduction (\cite{PlanasBielsa2004}) of
$(T^*S,\wt{\Om},\Hamc)$ with respect to the
$H$-action.

It follows immediately, for any dimension $n$, that the
ultimate reduced system on $T^*S^{n-1}$
(or rather on a fiber bundle $J_H^{-1}(\lam)/H_{\lam}\to T^*S^{n-1}$ -- see
Corollary~\ref{cor:4.2})
is Hamiltonian when the ball is
homogeneous.

For the non-homogeneous case, thanks to the
formula for $\wt{\Om}$ we can reprove the result of
\cite{BorisovMamaev} on Hamiltonization of the $3$-dimensional ball
in relatively simple geometric terms.
See Proposition~\ref{prop:ham}.
Hamiltonization of Chaplygin's ball for higher dimensions is still an
open problem.
It is, however, hoped that Theorem~\ref{prop:trunc}
can be of some help in this direction.
(After this paper was finished, important progress was made by
\cite{Jov09}.)

The proof of Hamiltonization of the 3-dimensional ball given in
\cite{BorisovMamaev2005}
relies on Chaplygin's reducing multiplier theorem.
This theorem
applies only to a certain kind of almost Hamiltonian systems with two degrees
of freedom, and states
that existence of a preserved measure is equivalent to
existence of a conformal factor.
An alternative method that
has been used to prove Hamiltonization of higher dimensional
non-holonomic systems is to explicitly establish an isomorphism with
a classical Hamiltonian system \cite{FJ04,Jov08,Jov09}. Our
approach  is valuable in that it is purely geometric, it does not have
an a-priori dimension restriction, and it ties together the work of \cite{EhlersKoiller}
and \cite{BorisovMamaev,BorisovMamaev2005}.

In Section~\ref{S:InternalSym} we study general $G$-Chaplygin systems
with internal symmetries. In this context we describe a reduction
procedure that is similar to reduction in stages in symplectic
geometry.
Section~\ref{S:Chaplygin} is used as a
motivation for doing so but can be read independently since all the
results are proved directly.
The set-up in this context is a generalization of the Chaplygin ball
described above. Thus $\pi: Q\to S$ is a $G$-principal fiber bundle with
connection one-form $\mathcal{A}$ and $\mu$ is an invariant metric.
The internal symmetries are modeled
by \emph{two} additional free and proper actions, called $l$ and $d$,
of the \emph{same} Lie group $H$ on $Q$ satisfying appropriate
compatibility conditions with regard to the connection $\A$ and the
metric $\mu$ and the projection $Q\to S$. The compression of the data
$(Q,\D=\A^{-1}(0),L=\by{1}{2}||\cdot ||_{\mu})$ and the induced
$H$-action on $T^*S$ are described.
From the non-holonomic Noether theorem it is concluded
that the standard \momap
$J_H:T^*S\to\ho^*$ with respect to the canonical symplectic form on
$T^*S$ is constant along flow lines of $\Xnh$. But $J_H$ need not
be the \momap associated to $\Omnh$. However,
we can replace $\Omnh$ with a non-degenerate and $H$-invariant two
form $\wt{\Om}$ which not only gives the correct dynamics,
$i(\Xnh)\wt{\Om}=i(\Xnh)\Omnh$, but also the desired \momap $J_H$,
$i(\zeta_Y)\wt{\Om} = \vv<dJ_H,Y>$ for all $Y\in\ho$.
This is accomplished via truncation with respect to a
choice of an auxiliary connection
$\sigma\in\Om^1(T^*S,\ho)$ on the principal bundle
$T^*S\toto(T^*S)/H$.
Again $\sigma$ is subject to the condition
that $\Xnh$ be horizontal. Such a $\sigma$ is shown to always exist over an
open sub-manifold of $T^*S$ which is invariant under the $H$-action and
the dynamics of $\Xnh$. This process of compression being followed by
reduction of internal symmetries has very much the flavor of
reduction in stages. Indeed, when $\A$ is the
mechanical connection associated to $\mu$ one can replace compression
followed by reduction by usual reduction in stages.

\section{The almost Hamiltonian setting and compression}\label{S:ham-setting}

A non-holonomic system is a triple $(Q,\D,L)$
where $Q$ is a configuration manifold,
$L: TQ\to\R$
is  a Lagrangian, and  $\D\subset TQ$ is
a smooth non-integrable distribution which is supposed to be
of constant rank. The equations of motion for a
curve $q(t)$ which should satisfy $q'\in\D$ are then stated
in terms of the Lagrange d'Alembert principle. We shall only
be concerned with Lagrangians of the form $L(q,v) =
\by{1}{2}\mu_q(v,v) - V(q)$ where $\mu$ is a Riemannian
metric on $Q$
and $V:Q\to\R$ is a potential.
In this
case there is also an (almost) Hamiltonian version (see
\cite{BS93,vanderschaft1994}, e.g.): continue to use the symbol $\mu$ to
denote the co-metric and consider the Hamiltonian $\Ham(q,p)
= \by{1}{2}\mu(p,p)+V(q)$.  Since $\D$ is of constant rank
there is a family of
independent one-forms $\phi^a\in\Om(Q)$ such that $\D$ is
the joint kernel of these.  In terms of coordinates
$(q^i,p_i)$ the equations of motion are
\[
 (q^i)' = \dd{p_i}{\mathcal{H}}
 \textup{ and }
 p_i' =
  -\dd{q^i}{\mathcal{H}} - \sum\lam_a\phi^a(\dd{q^i}{})
\]
where the $\lam_a$ are the Lagrange multipliers to be determined from
the supplementary condition that  $\mu(\phi^a,p)=0$.
With $X^{\mathcal{M}} := (q',p')$ we may
thus rephrase the equations as
\[
 i(X^{\mathcal{M}})\Om = d\Ham + \sum\lam_a\tau^*\phi^a
\] where $\Om=-d\theta$ is the canonical symplectic form on
$T^*Q$ and $\tau: T^*Q\to Q$ is the footpoint projection.
(The space
$\M\subset T^*Q$ which is a pseudonym for the distribution $\D$ will be
defined below.)

Roughly speaking, the process of writing the equations of
motion for a non-holonomic system in an almost Hamiltonian
way amounts to eliminating the Lagrange multipliers from the
equations of motion, and encoding the forces of constraint
in a bracket of functions (which fails the Jacobi identity) or a
(non-closed) two-form. Once this is
accomplished the constraints are satisfied
automatically. This process is developed below in terms of a
two-form.

\subsection{Chaplygin systems} \label{sub:G-chap}

Let $G$ be a Lie group that
acts freely, properly and by isometries on the Riemannian
manifold $(Q,\mu)$. A $G$-\caps{Chaplygin system} is a
non-holonomic system $(Q,L=\by{1}{2}||\cdot ||^{2}_{\mu},\D)$ that
has the property that $\D$ is a
principal connection on the principal bundle $Q\toto
Q/G$. Thus $\D$ is the kernel of a connection form $\A:
TQ\to\gu$. Notice that we do not require $\A$ to be the
mechanical connection associated to $\mu$. (In principle one
could also include a $G$-invariant function $V:Q\to\R$ but
we will not have use for this.)

We will now assume that $(Q,L,\D)$ is a
$G$-Chaplygin system and repeat some of the constructions
that are done in \cite{BS93}. In fact \cite{BS93} proceed in
greater generality. However, in the sequel we will only be
interested in Chaplygin systems whence the infinitesimal group orbit directions
form an exact complement to the distribution $\D$, and this
facilitates the development.

There is a sub-manifold
\[
 \M := \check{\mu}(\D)
\]
that corresponds to the constraint distribution, and the
inclusion will be denoted by $\iota: \M\hookto
T^*Q$. Clearly, $\M$ is invariant under the cotangent lifted
action by $G$, and there is an induced connection
$\iota^*\tau^*\A: T\M\to\gu$ on the principal bundle
$\M\toto \M/G$. Its horizontal space will be called
\[\C := (\iota^*\tau^*\A)^{-1}(0).
\] (This corresponds to the space $H$ in \cite{BS93}.)

\begin{theorem}[\cite{BS93}] The fiber-wise restriction of
$\iota^*\Om$ to $\C$ is non-degenerate.
\end{theorem}

Let us denote this restriction by $\Om^{\mathcal{C}}$. For
the simple reason that $\C$ is not the tangent space of any
manifold one cannot say that $\Om^{\mathcal{C}}$ is a
two-form. Nevertheless, morally it is this restriction
process that destroys the closedness property of
$\iota^*\Om$.  Since $X^{\mathcal{M}}$ is tangent to $\M$
and takes values in $\C$ one may thus rewrite the equations
of motion in the appealing format
\[
 i(X^{\mathcal{M}})\Om^{\mathcal{C}}
 =
 (d\Ham)^{\mathcal{C}}
\]
where $(d\Ham)^{\mathcal{C}}$ is the restriction of
$d\Ham$ to $\C$.

\subsection{Compression of $G$-Chaplygin systems}\label{S:CompressionIntrinsic}

In this section we review the compression of $G$-Chaplygin
systems from the Hamiltonian perspective. This will also
allow us to introduce some additional notation.
The original references are \cite{Koiller1992, BS93,
Koiller2002}. We shall follow \cite{Koiller2002} and use the
word compression instead of non-holonomic reduction.

Consider a $G$-Chaplygin system on a configuration manifold
$Q$ with constraint distribution $\D = \ker\A$ as defined
above.
Recall that $Q$ is endowed with the kinetic energy metric
$\mu$.  Let $\mu_0$ denote the induced metric on $S$ that
makes $\pi$ a Riemannian submersion. (To facilitate the
notation, we will sometimes tacitly identify tangent and
cotangent space of $Q$ and $S$ via their respective
metrics.)
Consider the orbit projection map
\[
 \rho: \M\toto \M/G.
\]
Using the respective metrics we can write $\rho$ as the
composition
\begin{equation}
\rho: \M\cong_{\mu}\D\overset{T\pi|\mathcal{D}}{\longto}
   TS\cong_{\mu_{0}}T^*S = \M/G.
\end{equation}
We may also associate a fiber-wise inverse to this
mapping which is given by the horizontal lift mapping
$\hlphi$ associated to $\A$.  (This inverse was called the
clock-wise diagram in \cite[Section~3.1]{EhlersKoiller}.)
As already noted above, $\wt{\A} := \iota^*\tau^*\A:
T\mathcal{M}\to \g$ defines a principal bundle connection
for $\rho$, whose horizontal spaces are given by $\C$.  (The
connection $\wt{\A}$ is the same as the one obtained in
\cite{EhlersKoiller} by differentiating the clock-wise
diagram.)

\begin{proposition}[Compression]
\begin{enumerate}[\up (1)]
The following are true.
\item
$\Om^{\mathcal{C}}$ descends to a
non-degenerate two-form $\Omnh$ on $T^*S$.
\item
$\Omnh = \Om_S - \vv<J_G\circ\textup{hl}^{\A},\tau_S^*K>$.
Here
$\Om_S=-d\theta_S$ is the canonical form on $T^*S$, $J_G$ is
the \momap of the cotangent lifted $G$-action on $T^*Q$,
$K\in\Om^2(S,\g)$ is the curvature form of $\A$, and
$\tau_S: T^*S\to S$ is the projection.
\item
The vector field $X^{\mathcal{M}}$ is $\rho$-related
to the vector field $\Xnh$ on $T^*S$ defined by
$i(X_{\textup{nh}})\Omnh = d\Hamc$ where the compressed
Hamiltonian, $\Hamc: T^*S = TS \to \R$ is defined by $\Hamc
:= \Ham\circ\textup{hl}^{\A}$, with $\textup{hl}^{\A}$
denoting the horizontal lift mapping.
\end{enumerate}
\end{proposition}

This result is well-known. It is contained in
\cite{BS93,Koiller1992,Koiller2002,EhlersKoiller}. Nevertheless
we include the following proof of the `$\vv<J,K>$-formula'
because it is slightly different in its flavor from those
that can be found in the literature. Note that our sign
in the $\vv<J,K>$-formula differs from that in
\cite{EhlersKoiller} since we are using a different
convention for the canonical exact symplectic form.
(For sake of brevity we will sometimes
write $\vv<J_G,K>$ instead of
$\vv<J_G\circ\hlphi,\tau_S^*K>$.)

\begin{proof}
We need to show that $\Om^{\mathcal{C}} =
\rho^*(\Om^S-\vv<J_G\circ\hlphi,\tau_S^*\;K>)$ on $\mathcal{C}$.
Work
locally, i.e., assume that $Q=S\times G$ is a direct product
and that $TS = S\times U$ is trivializable. Via right
trivialization we shall also identify $TG=G\times\gu$.  Thus
the connection is given by
\begin{equation}
\mathcal{A}: S\times U \times G\times\gu\longto \gu,
     (s,u,g,v)\longmapsto v+\mathcal{A}_{(s,g)}(u).
\end{equation} The horizontal space is
$\M=\D=\set{(s,u,g,-\mathcal{A}_{(s,g)}(u))}$. Let $X_1,X_2$
be vector-fields on $\M$ with values in $\C =
\wt{\mathcal{A}}^{-1}(0)$ that project to vector-fields
$\bar{X_1},\bar{X_2}\in\X(T^*S)$. Since $X_i\in \C$ we have
$X_i = (s_i',u_i',-\mathcal{A}(s_i'),v_i')$ for
$i=1,2$. Note also that $\bar{X_i} = (s_i',u_i')$ and
$T(\tau\circ\iota).X_i = (s_i',-\mathcal{A}(s_i'))\in
T(S\times G)$.

Let $\psi: T\M\to\C$ denote the horizontal projection
associated to $\wt{\A}$.

Claim: $\theta^Q\circ\psi = \rho^*\theta^S$ where
$\theta^Q$, $\theta^S$ denote the respective canonical
one-forms.  Indeed, for $(q,p)\in\M$ and $\rho(q,p) =
\rho(s,g,u,-\mathcal{A}_{(s,g)}(u)) = (s,u)$ we find
\[\theta^Q X_2(q,p) =
\mu_q\big((u,-\mathcal{A}(u)),(s_2',-\mathcal{A}(s_2'))\big)
= (\mu_0)_s(u,s_2') = \theta^S\bar{X_2}(s,u)
\] which shows the identity for all vector-fields where it is
non-trivial.  Therefore,
\[X_1.\theta^Q X_2 = d((\theta^S \bar{X_2})\circ\rho).X_1 =
(d(\theta^S \bar{X_2}).d\rho(X_1))\circ\rho =
(\bar{X_1}.\theta^S \bar{X_2})\circ\rho.
\] Since
$\psi+\zeta\circ\wt{\mathcal{A}}=\id_{T\mathcal{M}}$ where
$\zeta$ is the fundamental vector-field mapping of the
$G$-action on $\M$, it follows that
\begin{align*}
  \Om^{\mathcal{M}}(X_1,X_2)
  &=
  -X_1.\theta^Q X_2 + X_2.\theta^Q X_1
    + \theta^Q\psi[X_1,X_2]
    + \theta^Q(\zeta\circ\wt{\mathcal{A}})[X_1,X_2]\\
  &=
  -\rho^*(\bar{X_1}.\theta^S\bar{X_2}
    - \bar{X_2}.\theta^S\bar{X_1})
    + \rho^*(\theta^S[\bar{X_1},\bar{X_2}])
    - \vv<J_G,\tau^*K>(\bar{X_1},\bar{X_2})\\
  &=
  -\rho^*(d\theta^S)(X_1,X_2)
  - \rho^*\vv<J_G\circ\hlphi,\tau_S^*K>(X_1,X_2).
\end{align*}
For the middle equation we used the following identity. Let
$\wt{K}\in\Om^2(\mathcal{M},\gu)$ denote the curvature form
associated to $\wt{\mathcal{A}}$. Then
$\theta^Q(\zeta\circ\wt{\mathcal{A}})[X_1,X_2]
 =
 \vv<J_G,\wt{\mathcal{A}}[X_1,X_2]>
 =
 \vv<J_G,-\wt{K}(X_1,X_2)>
 =
 \vv<J_G,-K\circ\Lam^2 T\tau(\bar{X_1},\bar{X_2})>
 $.
\end{proof}

We collect the  compressed data to a triple
$(T^*S,\Omnh,\Hamc)$ and refer to it as the compressed
system. Let us also note explicitly that the equations of
motion are now given by the almost Hamiltonian form
\begin{equation}
 i(\Xnh)\Omnh = d\Hamc
\end{equation}
whence the constraints have been successfully encoded in the
two-form structure.
In general, $\Omnh$ is an almost symplectic form, that is, it is
non-degenerate and non-closed. Thus we refer to the compressed system
$(T^*S,\Omnh,\Hamc)$ as an almost Hamiltonian system. However, there
are non-integrable distributions which do give rise to forms $\Omnh$
that are closed. This is simply so because compression is a
generalization of symplectic reduction at the $0$-level of simple
mechanical systems. Consider, e.g., the homogeneous Veselova
system of \cite{Ves88}. For this system the configuration space $Q$ is $\SO(3)$, the
group $G$ is $S^1$, and the distribution $\D$ is the horizontal space
of the mechanical connection associated to the standard biinvariant metric
on $\SO(3)$.
(The constraints are conserved quantities of the unconstrained system.)
Thus compression and symplectic reduction at
$0\in\gu^*=\R$ agree.

\section{Reduction of internal symmetries via truncation}
\label{S:InternalSym}

We continue notation and assumptions from
Section~\ref{S:CompressionIntrinsic}.  Thus $\pi: Q\to S$ is
a $G$-principal fiber bundle with connection form
$\mathcal{A}$.  Additionally we assume that there is a Lie
group $H$ which acts on $S$, through a linear representation
on $\gu$, and by two different actions, $l$ and $d$, on
$Q$. More precisely we require that
\begin{itemize}
\item
$\pi: Q\to S$ is $l$- and $d$-equivariant;
\item
$\A: TQ\to\gu$ is $d$-equivariant;
\item
$l$ acts by internal symmetries, that is $\A.\zeta^l_Y
= 0$ for all $Y\in\ho$.
\end{itemize}
The metric $\mu$ on $Q$ is now supposed to be
$l$-, $d$-, and $G$-invariant.
This is the abstraction of the situation encountered in
Section~\ref{S:Chaplygin}.

In non-holonomic mechanics the relationship between symmetries and
conserved quantities is not obvious. (See \cite{BKMM}.)
While the momentum map for an external symmetry group (the
$G$- and $d$-actions) is generally not constant during the
motion, the momentum map associated to a internal symmetry
(the $l$-action) is. This is the non-holonomic version of
Noether's theorem which we state for further reference in
the following theorem that can be found in \cite{ArnoldIII}.

\begin{theorem}\label{T:NoetherNonho} Let $H$ be an internal
symmetry group of a non-holonomic system.
Then the momentum map $J_H: T^*Q\to\ho^*$ is
constant during the motion.
\end{theorem}

By an internal symmetry of $(Q,\D,L)$ we mean an action by a Lie group
$H$ on $Q$ such that $L$ is $H$-invariant and $\zeta_Y\in\D$ for all
$Y\in\ho$.
However, $\D$ is not required to be $H$-invariant.

\subsection{Compression in the presence of internal symmetries}

Via the metric we identify $TQ$ and $T^*Q$ and
the horizontal bundle $\mathcal{D} = \mathcal{A}^{-1}(0)$ is
identified with its image $\mathcal{M}\subset T^*Q$.  Let
$\mu_0$ denote the induced metric on $S$.  As in
Section~\ref{S:CompressionIntrinsic} we denote the
compressed Hamiltonian by $\Hamc := \Ham\circ\hlphi$ where
$\hlphi: TS\to\mathcal{D}$ is the horizontal lift mapping.
Recall also the projection $\rho: \M\to\M/G=T^*S$.  The
following describes how the internal symmetries descend to
the compressed system $(T^*S,\Omnh,\Hamc)$.

\begin{proposition}
\label{P:compressioninternalsymmetries} The following are
true.
\begin{enumerate}[\up (1)]
\item The $H$-action $d$ restricts to $\M$, and $\rho$ is
equivariant with respect to the cotangent lifted $H$-action on
$T^*S$.
\item $\Omnh$ is $H$-invariant.
\item $\Hamc$ and $\Xnh$ are $H$-invariant.
\item $J_H = (J_l|\M)\circ\hlphi$ where $J_H$ is the
standard \momap of the cotangent lifted $H$-action on
$(T^*S,\Om^S)$ and $J_l$ is the standard \momap of the
lifted $l$-action on $(T^*Q,\Om^Q)$.
\item $dJ_H.\Xnh = 0$.
\end{enumerate}
\end{proposition}

Note that $l$ does not necessarily restrict to an action on $\M$ and
$J_d$ (the \momap of the $d$-action) does not factor to
$J_H$.

\begin{proof}
(1) This is clear from the assumptions.

(2)
Note that $\Om^S$ is clearly
invariant. Further, $J_G$ is $H$-equivariant with respect to
the $d$-action since $\zeta^G: \gu\to TQ$ is $d$-equivariant
by assumption, and the same is true for $\hlphi$ and
$K$. That is, $h^*\vv<J_G\circ\hlphi,K> =
\vv<h^*(J_G\circ\hlphi),h.K> = \vv<J_G\circ\hlphi,K>$ for
all $h\in H$, since $h^*\alpha = \alpha\circ h^{-1}$ for
$\alpha\in \gu^*$.

(3)
Since $\hlphi$ is $H$-equivariant for the $d$-action and
$\Ham$ is invariant the first point is clear.  For the
second we use that also $\Omnh$ is $H$-invariant whence
\[h^*\Xnh = h^*((\check{\Omnh})^{-1}d\Hamc) =
((h^*\Omnh)^{\check{}})^{-1}h^*d\Hamc = \Xnh
\] for $h\in H$.

(4)
Let $Y\in\ho$ and $(s,u)\in T^*S=TS$, then
$\hlphi_q(\zeta_Y(s)) = \zeta_Y^l(q)$ implies that
\[\vv<J_H(s,u),Y> = (\mu_0)_s(u,\zeta_Y(s)) =
\mu_q(\hlphi_q(u),\zeta_Y^l(q)) = \vv<J_l(\hlphi(s,u)),Y>.
\]

(5)
This follows from the previous point and
Theorem~\ref{T:NoetherNonho}.
\end{proof}

\subsection{Truncation}\label{S:truncation}
By the above proposition the compressed system\break
$(T^*S,\Omnh,\Xnh)$ is $H$-invariant and $J_H$ is a conserved
quantity. Thus one would be tempted to do almost Hamiltonian
reduction.  (See \cite{PlanasBielsa2004}.)  However, $J_H$ is not
the \momap corresponding to $\Omnh$, that is, for $Y\in\ho$
\[
 i(\zeta_Y)\Omnh \neq d\vv<J_H,Y>
\]
in general.  Thus the restriction of $\Omnh$ to a level
set of $J_H$ will not be a horizontal form in general
whence it does not factor to a reduced form on the `would be' almost
symplectic quotient. (This is the situation for the Chaplygin ball
problem treated in Section~\ref{S:Chaplygin}.)
To
remedy the situation we truncate the $\vv<J_G,K>$-term thus
changing $\Omnh$ in a certain way that does not affect the
equations of motion.
In effect,
we will replace $(T^*S,\Omnh,\Hamc)$ by a different almost
Hamiltonian system $(T^*S,\wt{\Om},\Hamc)$ which has the same
dynamics given by $\Xnh$.

To motivate the construction notice that the obstruction to
horizontality of the restriction of $\Omnh$ to a level set of $J_H$ is
just
\[
 i(\zeta_Y)\Omnh - d\vv<J_H,Y>
 =
 i(\zeta_Y)\Omnh - i(\zeta_Y)\Om^S
 =
 -i(\zeta_Y)\vv<J,K>.
\]
So the vertical directions are problematic.
On the other hand, we have for the dynamics
\[
 d\Hamc = i(\Xnh)\Om^S - i(\Xnh)\vv<J,K>.
\]
Thus we need to invent a device whereby we make the $\vv<J,K>$-term
vanish upon insertion of vertical vectors while it remains unchanged
when contracted with $\Xnh$. In particular we need a way to distinguish
$\Xnh$ from vertical directions. This calls for a connection such that
$\Xnh$ is horizontal.

\begin{theorem}\label{thm:trunc}
Suppose there is a
connection $\sigma\in\Om(T^*S,\ho)$ of the principal bundle
$T^*S\toto(T^*S)/H$ that satisfies $\sigma\Xnh = 0$.  (See
Proposition~\ref{P:existence-of-connct}.)  Let $\chi:
TT^*S\to TT^*S$ denote the horizontal projection associated
to $\sigma$.  Then the truncated form
\[
  \wt{\Om}
  :=
  \Om^S - \vv<J_G\circ\hlphi,\tau_S^*K>\circ\Lam^2\chi
\]
has the following properties.
\begin{enumerate}[\up (1)]
\item
It is non-degenerate.
\item
It is $H$-invariant.
\item
$i(\Xnh)\wt{\Om} = d\Hamc$.
\item
$i(\zeta_Y)\wt{\Om} = d\vv<J_H,Y>$ for all $Y\in\ho$.
\end{enumerate}
\end{theorem}

\begin{proof} Properties (1) and (2) are immediate. (Use
that $\chi$ is $H$-equivariant for the second.)

(3)
We need to show that $\wt{\Om}(\Xnh,X) = \Omnh(\Xnh,X)$
for all
$X\in\X(T^*S)$. If $X$ is horizontal then this is
obvious. Suppose $X$ is vertical, that is $X=\zeta_Y$ for
some $Y\in\ho$. But then we have
\[
 \vv<J_G,K>(\Xnh,\zeta_Y) = 0;
\]
this follows because
\[
 \Omnh(\Xnh,\zeta_Y) = d\Hamc.\zeta_Y = 0
\]
by $H$-invariance of $\Hamc$, and
\[
 \Om^S(\Xnh,\zeta_Y)
 = -\vv<dJ_H.\Xnh,Y>
 = 0
\]
by conservation of $J_H$.  Thus $\wt{\Om}(\Xnh,\zeta_Y) =
\Om^S(\Xnh,\zeta_Y) = \Omnh(\Xnh,\zeta_Y)$.

(4)
This is true since the truncated $\vv<J_G,K>$-term vanishes, by
construction,
on vertical vectors and $J_H$ is the canonical \momap.
\end{proof}

Note that the above proof relies on both decisive features of an
almost Hamiltonian system with symmetries: it uses invariance of the
Hamiltonian as well as the conserved quantity.

When $\A$ is the mechanical connection on $Q\toto S$ associated to
the metric $\mu$ then compression equals symplectic reduction at
$0$. Thus $\Omnh=\Om^S$ is a true symplectic form in this case and the
$H$-action is Hamiltonian with \momap $J_H$. Obviously, this is
compatible with the truncation procedure in the trivial sense. Thus we
recover sympletic reduction in stages.

Of course, there may also be a connection $\wt{\sigma}\in\Om^1(Q,\ho)$
with the property that $\Xnh\in\ker\tau^*\wt{\sigma}$. If this is the
case then one can replace $\sigma$ in the theorem by $\tau^*\wt{\sigma}$
but in general this seems to be too much to ask for. The analog of
Proposition~\ref{P:existence-of-connct} does not hold.

Thus to describe the dynamics of $\Xnh$ we may deal with the
system $(T^*S,\wt{\Om},\Hamc)$ which has the advantage that
it not only admits $H$ as a symmetry group but also produces
the desired \momap. Now one can perform Hamiltonian reduction
(\cite{PlanasBielsa2004}) with respect to the non-closed
form $\wt{\Om}$.

Now we address the question of existence of the auxiliary connection $\sigma$
needed for truncation.
Consider the vertical space $\ver(H)$ of the lifted
$H$-action on $T^*S$. Define the sets
\[
\mathcal{E} = \Xnh^{-1}(\ver(H))
\textup{ and }
\mathcal{U} = (T^*S)\setminus\mathcal{E}
\]
and note that $\mathcal{U}$ is an open
sub-manifold while $\mathcal{E}$ is the set of relative equilibria.

\begin{proposition}
\label{P:existence-of-connct} The following are true.
\begin{enumerate}[\up (1)]
\item
$\mathcal{U}$ and $\mathcal{E}$ are $H$-invariant and
invariant under the dynamics of $\Xnh$.
\item
On $\mathcal{U}$ there is a connection $\sigma$ such
that $\Xnh$ is horizontal.
\end{enumerate}
\end{proposition}

\begin{proof}
(1)
It suffices to show the assertions for
$\mathcal{E}$. Invariance under the $H$-action is clear since $\Xnh$ and
$\ver(H)$ are $H$-invariant.  Fix $(s,u)\in\mathcal{E}$ such
that $\Xnh(s,u)=\zeta_Y(s,u)$ for some $Y\in\ho$, and
consider the curve $c(t) = \exp(tY).(s,u)$. By $H$-invariance
the curve stays in $\mathcal{E}$.  We show that it is an
integral curve:
\begin{align*}
c'(t)
 = \zeta_Y(\exp(tY).(s,u))
 = \exp(tY).\zeta_Y(s,u)
 = \exp(tY).\Xnh(s,u)
 = \Xnh(c(t)).
\end{align*}
(2)
Consider the $H$-invariant sub-bundle of $T\mathcal{U}$ given
by $\mathcal{F} = \R\Xnh\oplus\ver(H)$ where $\ver(H)$ is
the vertical space of the induced $H$-action on
$\mathcal{U}$.  Take an $H$-invariant metric on
$\mathcal{U}$.  Such a metric always exists since the action
is proper. Now define a horizontal sub-bundle by
$\hor(\sigma) = \R\Xnh\oplus\mathcal{F}^{\perp}$ where the
orthogonal is taken with respect to the metric.  By
construction $\hor(\sigma)$ is an $H$-invariant complement
of the vertical space of the $H$-action on $\mathcal{U}$.
\end{proof}

Consider a point $(s,u)\in\mathcal{E}$. Clearly
$s' = T\tau_S.\Xnh(s,u) = u$
where we identify again $T^*S=TS$ via
$\mu_0$.  Since $(s,u)\in\mathcal{E}$ there is a $Y\in\ho$
such that $\Xnh(s,u) = \zeta^{TS}_Y(s,u)$ whence $u=\zeta^l_Y(s)$.
In the physical applications we have in mind it is true that
$\mu(\zeta^l_Y,\zeta^G_v)=0$.  Therefore, $\vv<J_G(s,u),v> =
\mu_q(\hlphi(u),\zeta^G_v) =
-\mu_q(\zeta^l_Y,\zeta_v^G) = 0$
whence the $\vv<J_G,K>$-term vanishes upon restriction to
$\mathcal{E}$.  We can thus understand the dynamics on
$(T^*S,\Omnh,\Hamc)$ by treating
$(\mathcal{E},\Om^S|\mathcal{E},\Hamc|\mathcal{E})$ and
$(\mathcal{U},\wt{\Om},\Hamc|\mathcal{U})$ as individual
problems.

\section{Example: Chaplygin's rolling ball}\label{S:Chaplygin}

The $n$-dimensional Chaplygin ball ($n\ge3$) concerns a rigid ball that rolls on an
$n-1$-dimensional table without slipping and whose geometric center
coincides with its center of mass. The mass distribution is not
assumed to be homogeneous.

By adjusting the units appropriately we assume the radius and the mass
of the ball both equal to $1$.  It is convenient to identify the table with
$\R^{n-1}\times\set{-1}$ whence the motion of the center of the ball is
given by a curve $(x(t),0)\in\R^{n-1}\times\set{0}$.\footnote{We will
  write row vectors but treat them as column vectors.}
Let $e_1,\dots,e_n$ denote the standard basis of $\R^n$. The
orientation of the ball at time $t_0$ is determined by a
unique element $s(t_0)\in\SO(n)=:S$ that relates this basis to a
moving frame that is attached to the center of the ball and rotates
with it. Thus the configuration space of the system is
\[
 Q:=S\times\R^{n-1}.
\]
Consider a fixed marked point $b$ on the surface of the ball. The
motion of this point is described by the curve
\[
 z(t) = (x(t),0)+s(t).b.
\]
The constraint of rolling without slipping is that the
velocity of the contact point is $0$. For the contact point at time
$t_0$ we have that $s(t_0).b=-e_n$ whence $z'(t_0)=0$ implies that
\[
 (x'(t_0),0) = s'(t_0)s(t_0)^{-1}.e_n.
\]
We put
$s'(t_0)s(t_0)^{-1} = \wt{u}\in\so(n)_{\textup{R}}$
where $\so(n)_{\textup{R}}$ is identified with the Lie algebra of
right invariant vector fields on $S$.
In other words, the constraints are satisfied iff
\[
 (s's^{-1},x')
  \in\wt{\D}
 := \set{(\wt{u},x')\in\so(n)_{\textup{R}}\times\R^{n-1}:
          \wt{u}.e_n = (x',0)}.
\]
Thus the set of allowed motions is described by the condition that the
velocities in the right trivialization (space frame) belong to $\wt{\D}$.
If we define
\[
 \wt{\A}:
 \so(n)_{\textup{R}}\longto\R^{n-1},\;
 \wt{u}\longmapsto\wt{u}.e_n
 \longmapsto -\sum_{a=1}^{n-1}\vv<e_a,\wt{u}.e_n>_E \,e_a,
\]
where $\vv<.,.>_E$ denotes the standard inner product,
then
\[
 \wt{\D}
 =
 \set{(s,\wt{u},x,-\wt{\A}(\wt{u}))}
 \subset
 S\times\so(n)_{\textup{R}}\times T\R^{n-1}.
\]
The sign in the definition of $\wt{\mathcal{A}}$ is included so that the
associated horizontal subspace (see below) can be written in the usual
way.
Let
\[
 H := \set{h\in S: h.e_n=e_n}
\]
with Lie algebra $\ho$, let
$\vv<.,.>$ denote the Killing form, and let $\ho\oplus\ho^{\bot}$ be
the corresponding decomposition of $\so(n)$.
When appropriate we will identify $H=\SO(n-1)$ and consider it as
acting on $\R^{n-1}$.
In terms of matrix
notation $\ho^{\bot}$ corresponds to the subspace of matrices that
have non-zero entries only in the last column and row.
Let
\[
 Y_{\alpha},\;
 Z_a,\;
 \alpha=1,\dots,\dim\ho,\;
 a=1,\dots,n-1
\]
denote an
orthonormal basis that is adapted to this decomposition. Then, if the
basis is ordered and oriented in the right way, we may
write
$\wt{\A} = -\sum_a \vv<Z_a,.>e_a$.
It
will be convenient to work with the left trivialization (body frame).
From now on we trivialize $TS = S\times\so(n)$ via the left
trivialization. Consider the $\R^{n-1}$-valued one-form on $S$ defined by
\[
 \A: TS = S\times\so(n)\longto\R^{n-1},\;
 (s,u)\longmapsto\Ad(s).u = \wt{u}\longmapsto\wt{\A}(\wt{u}).
\]
Via right multiplication we extend the
basis $Y_{\alpha},Z_{a}$ to a frame on $S$:
\begin{equation}\label{e:rhos}
  \xi_{\alpha}(s)
   := \zeta^l_{Y_{\alpha}}(s)
    = \Ad(s^{-1})Y_{\alpha}
  \textup{ and }
  \zeta_a(s)
   := \zeta^l_{Z_{a}}(s)
    = \Ad(s^{-1})Z_{a}.
\end{equation}
The corresponding co-frame will be called
$\rho^{\alpha},\eta^a$. We shall stick to the convention of
using lower case Greek letters $\alpha,\beta,\gamma$ to
refer to $Y_{\alpha}$'s and lower Latins $a,b,c$ for
$Z_a$'s.
In this frame the form $\A\in\Om^{1}(S,\R^{n-1})$ reads
\[
 \A = -\sum_{a=1}^{n-1}\eta^a e_{a}.
\]
For $n=3$ one can get the same formula for $\A$ as in
\cite{EhlersKoiller} by inserting appropriate signs which corresponds
to rearranging the basis. All such choices cancel out in the
subsequent.

Let $\mathbb{I}$ denote the inertia tensor that describes the mass
distribution of the ball. Then the appropriate metric on $Q$ is the product metric
$\mu=\vv<\mathbb{I}.,.>+\vv<.,.>_E$ and the Lagrangian of the system is the
kinetic energy function associated to $\mu$. Thus the Chaplygin ball is
the non-holonomic system described by the triple
\begin{equation}\label{e:ball}
 (Q,\D,\by{1}{2}||\cdot ||_{\mu})
\end{equation}
where $\D$ is the sub-bundle defined by
\[
 \D
 =
 \set{(s,u,x,-\A_s(u))}
 \subset
 S\times\so(n)\times T\R^{n-1}.
\]
Notice that the kinetic energy of the system is left invariant while the
distribution that describes the constraints is right invariant.
The system (\ref{e:ball}) is invariant under the Lie group
action given by addition of $\R^{n-1}$ on the $\R^{n-1}$-factor of
$Q$.
Clearly, $\D$ defines a connection on the principal bundle $\R^{n-1}\hookto
Q\toto S$ with connection form $\A\in\Om^{1}(S,\R^{n-1})$.
Thus (\ref{e:ball}) is a $G$-Chaplygin system in the sense of
Section~\ref{S:ham-setting} with $G=\R^{n-1}$.
Moreover, $\A$ has the following properties.
\begin{enumerate}[\up (1)]
\item
$\mathcal{A}(hs,u) = h\A(s,u)$ for all $h\in H$ and
$(s,u)\in TS$, that is, $\D$ is invariant under the diagonal $H$-action.
\item
$H$ acts by internal symmetries, that is,
$\A.\zeta^l_Y=0$ for all $Y\in\ho$.
\item
$\A(sg^{-1},\Ad(g)u) = \A(s,u)$ for all $g\in S$
and $(s,u)\in TS$.
\end{enumerate}
These properties have a physical meaning. Property (1) says that the
constraints are invariant under simultaneous rotation of the space
frame and the ball about the vertical axis. The second says that
rotation of the ball about the vertical axis is an allowed motion. The
third states that the system is invariant with respect to rotations of
the space frame. Notice also that properties (1) and (2) correspond to
the compatibility conditions stated at the beginning of
Section~\ref{S:InternalSym}.

\subsection{The compressed system}

As stated above the Lagrangian $\mathcal{L}$ of the system is the kinetic energy
associated to $\mu$. This Lagrangian is non-degenerate whence we are
in the situation of Section~\ref{S:ham-setting}, and we will denote
the corresponding Hamiltonian by $\Ham$. For ease of notation we will
write $V:=\R^{n-1}$.

Let
$\Om^Q$ be the canonical symplectic form on $T^*Q$.
In accordance with Section~\ref{S:CompressionIntrinsic}
we
describe now the compression of the system
$(Q,\mathcal{L},\D=\set{(s,u;x,-\A_s(u))})$.
We will henceforth identify
\begin{equation}
 T^*S = TS = \D/V
\end{equation}
via the induced metric $\mu_0$.
The compressed
Hamiltonian reads
\[
 \Hamc(s,u) = \by{1}{2}\vv<u,\mathbb{I}u>+\by{1}{2}\vv<\A_s(u),\A_s(u)>,
\]
and note that $\Hamc$ is invariant under the induced
$H$-action on $TS$.
(This action has various equivalent descriptions -- see
Proposition~\ref{P:compressioninternalsymmetries}.)
The Hamiltonian $\Hamc$ is the sum of a left-
and a right-invariant factor. Systems of this type are
sometimes called $L+R$-systems. See \cite{Fed99}.

According to Section~\ref{S:CompressionIntrinsic} the
compressed almost symplectic form on $TS$ is of
the form
\[
 \Omnh
 = \Om^S - \vv<J_V\circ\textup{hl}^{\mathcal{A}},d\A>
 = \Om^S + \vv<\A,d\A>.
\]
It will be convenient to introduce the following set of
functions on $TS$:
\[
  l_{\alpha}(s,u) = \rho^{\alpha}_s(u),\;
    \wt{l}_{\alpha}(s,u) = l_{\alpha}(s,\mathbb{I}u),
    \textup{ and }
        g_a(s,u) = \eta^a_s(u),\;
    \wt{g}_a(s,u) = g_a(s,\mathbb{I}u)
\]
where $\rho^{\alpha},\eta^a$ denotes the co-frame associated to
(\ref{e:rhos}).
These functions have a physical meaning; $l_{\alpha},g_a$
are the components of angular velocity in the space frame
and $\wt{l}_{\alpha},\wt{g}_a$ are those of angular momentum
about the center of mass also in the space frame.
We may thus write the canonical symplectic form as
\begin{equation}\label{e:Om^S}
 \Om^S
  = -d\big(\sum \wt{l}_{\alpha}\rho^{\alpha}
    + \sum(\wt{g}_a+g_a)\eta^a\big).
\end{equation}
(Remember that the identification of $TS$ with its dual is
via $\mu_0$.)
The formulas
\begin{equation}\label{e:drho}
    d\rho^{\alpha} = \by{1}{2}\sum
             c^{\alpha}_{\beta\gamma}\rho^{\beta}\wedge\rho^{\gamma}
             + \by{1}{2}\sum
             c^{\alpha}_{ab}\eta^a\wedge\eta^b
    \textup{ and }
    d\eta^a = \sum c^a_{\beta b}\rho^{\beta}\wedge\eta^b
\end{equation}
will be used very often; here the summation is over
repeated indices and $c^._{..}$ are the structure constants.
The compressed form thus becomes
\begin{equation}\label{e:Omnh}
 \Omnh
  = \Om^S + \sum g_a c^a_{\beta b}\rho^{\beta}\wedge\eta^b.
\end{equation}
Via the trivialization we write the non-holonomic
vector-field $\Xnh = (\check{\Omnh})^{-1}d\Ham$ on $TS$ as
\[
 \Xnh(s,u) = (s'(s,u),u'(s,u)) \in\so(n)\times\so(n).
\]
Using right
invariant vector fields we thus have that
\begin{equation}\label{e:s'}
    s' = \sum l_{\alpha}\xi_{\alpha} + \sum g_a\zeta_a.
\end{equation}
This is just the first half of Hamilton's equations which says that
$s' = u$.

According to
Proposition~\ref{P:compressioninternalsymmetries} the \momap
associated to the $l$-action compresses to the standard
\momap
\[
J_H: TS\longto\ho^* =_{\vv<.,.>} \ho,\text{ }
(s,u)\longmapsto\sum\wt{l}_{\alpha}(s,u)Y_{\alpha}
\]
with
respect to the lifted $H$-action on $(TS,\Om^S)$.
Furthermore, we have the conservation law $dJ_H.\Xnh=0$.

\subsection{Truncation}\label{S:trunc-ball}
We are now in the situation of Section~\ref{S:truncation}. Namely
one can verify that
the
conserved quantity $J_H$ is \emph{not} the \momap with respect to
$\Omnh$. Thus $\Omnh$ does \emph{not} factor to a two form on
quotients of the type $J_H^{-1}(\lam)/H_{\lam}$. Therefore, we need to
change $\Omnh$ in a certain way.

According to Theorem~\ref{thm:trunc} we have to find a connection $\sigma$ on the
principal bundle $TS\toto (TS)/H$ such that $\Xnh$ is horizontal. This means that
$\chi(\Xnh)=\Xnh$ where $\chi: T(TS)\to T(TS)$ is the associated
horizontal projection. Let us also trivialize
\[
 T(TS) = T(S\times\so(n)) = TS\times T\so(n)
 = S\times\so(n) \times \so(n)\times\so(n)
\]
via left-multiplication.
Then $\sigma$ has to be of the form
\[
 \sigma
 =
 \big(
  \sum(\rho^{\alpha} + f_a^{\alpha}\eta^a)\otimes\xi_{\alpha} , 0
 \big)
\]
where the $f_a^{\alpha} = f_a^{\alpha}(s,u)$ are unknown functions.
Thus
\[
 \chi
 =
 \big(
  -\sum f_a^{\alpha}\eta^a \otimes\xi_{\alpha} + \sum\eta^a\otimes\xi_a , \id_{\mathfrak{so}(n)}
 \big).
\]
The condition that $\Xnh$ be horizontal becomes
\begin{equation}\label{e:def-equ}
 l_{\alpha} = -\sum f_a^{\alpha}g_a.
\end{equation}
In accordance with Proposition~\ref{P:existence-of-connct}
this is solvable on the complement
of the set
\[
 \mathcal{E}
 = (\Xnh)^{-1}(\ho\times\set{0})
 = \set{g_a = 0}
 \subset TS.
\]
However, for convenience of exposition we restrict to the somewhat smaller set
$\mathcal{U}' := \set{(s,u): g_a(s,u) \neq 0 }\subset\mathcal{E}^c$.
One particular choice for $\chi$ that solves equation~(\ref{e:def-equ}) is
\[
 \chi
 =
 \big(
   \by{1}{n-1}\sum \by{l_{\alpha}}{g_a}\eta^a \otimes\xi_{\alpha} + \sum\eta^a\otimes\xi_a , \id_{\mathfrak{so}(n)}
 \big).
\]
(In fact, since $\chi$ has to be $H$-equivariant one does not have so much freedom here.
Choosing
$f_a^{\alpha} = -\by{l_{\alpha}}{g_a}\delta_{1a}$ solves (\ref{e:def-equ}) but
does not
yield an equivariant $\chi$, for example.)

Our strategy will now be to truncate $\vv<J,K>$ using $\chi$. This
truncation will be well-defined
on $\mathcal{U}'$ only. However, it will be obvious how to extend the result
to a two-form on the whole space.

When $n=3$ the truncated form is especially easy to compute. (For notational
reasons we make the convention that
$a=1,2$ and $\alpha=3$ whence the basis receives the appelation $Z_1,Z_2,Y_3$.)
Indeed,
\[
 \vv<J,K>(\chi\zeta_1,\chi\zeta_2)
 = -\sum g_a c^a_{\alpha b}\rho^{\alpha}\wedge\eta^b(\chi\zeta_1,\chi\zeta_2)
 = l_3
 = \by{1}{2}\sum c^{\alpha}_{ab}l_{\alpha}\eta^a\wedge\eta^b(\zeta_1,\zeta_2).
\]
Thus we can replace $\vv<J,K>$ with the semi-basic two form
\[
 \wt{\vv<J,K>} := \by{1}{2}\sum c^{\alpha}_{ab}l_{\alpha}\eta^a\wedge\eta^b
\]
which is obviously
well-defined on the whole space and also makes sense for $n>3$.

However, notice that $\wt{\vv<J,K>} \neq \vv<J,K>\circ\Lam^2\chi$ for $n>3$.
They agree only on a set of measure zero and the truncated form is not
defined on the whole space.
The point is that the contraction with
$\Xnh$ does not see this difference whence we may use $\wt{\vv<J,K>}$.

Now we notice that $\wt{\vv<J,K>}$ can be written in terms of
well known geometric objects.
Namely, let
\[
 \wt{\Om} :=
 \Om^S - \wt{\vv<J,K>}
 =
 \Om^S - \vv<L,\curv^{\om}>
\]
where $L = \sum l_{\alpha}Y_{\alpha}$ and
$\curv^{\om}\in\Om^2(S,\ho)$ is the curvature of the
standard $H$-connection
$\om
 =
 \sum \rho^{\alpha}\otimes Y_{\alpha}$.
Thus $L$ is the one-form $\om$ viewed as a function
$TS\to\ho$.

\begin{theorem}\label{prop:trunc}
The system
$(TS,\wt{\Om},\Hamc)$ has the following properties.
\begin{enumerate}[\up (1)]
\item
$\wt{\Om}$ is almost symplectic and $H$-invariant;
\item
$i(\Xnh)\wt{\Om} = d\Hamc$;
\item
$J_H$ is a \momap of the $H$-action
on $(TS,\wt{\Om})$.
\end{enumerate}
\end{theorem}

\begin{proof} Clearly $\wt{\Om}$ is non-degenerate, and the
term $\vv<L,\curv^{\om}>$ is $H$-invariant because
ambiguities in the pairing cancel out.
The third assertion is also obvious since
\[
 i(\xi_{\alpha})\wt{\Om}
  = i(\xi_{\alpha})\Om^S
  = \vv<dJ_H,Y_{\alpha}>.
\]
Thus it remains to show that
\[
 i(\Xnh)\vv<J,K> = i(\Xnh)\vv<L,\curv^{\om}>.
\]
Notice that $\vv<J,K>(\Xnh,\xi_{\alpha}) = 0$ by the
proof of Theorem~\ref{thm:trunc}.  Equating on $\om$-horizontal
vector fields and using formula (\ref{e:s'}) for $T\tau.\Xnh$ yields
\[
 \vv<J,K>(\Xnh,\zeta_c)
 = -\sum g_a c^a_{\alpha b}\rho^{\alpha}\wedge\eta^b(\Xnh,\zeta_c)
 = -\sum c^a_{\alpha c}g_a l_{\alpha}.
\]
On the other hand:
\[
 \vv<L,\curv^{\om}>(\Xnh,\zeta_c)
 = \sum_{a<b}l_{\alpha}c^{\alpha}_{ab}\eta^a\wedge\eta^b(\Xnh,\zeta_c)
 = \sum c^{\alpha}_{ac} g_a l_{\alpha}
\]
where we have used that
\begin{align*}
\curv^{\om}
 &=
  d\om-\by{1}{2}[\om,\om]
  =
  \sum d\rho^{\alpha}Y_{\alpha}
    - \sum_{\alpha<\beta}\rho^{\alpha}\wedge\rho^{\beta}c^{\gamma}_{\alpha\beta}Y_{\gamma}\\
 &=
 \sum_{\beta<\gamma,b<c}(c^{\alpha}_{\beta\gamma}\rho^{\beta}\wedge\rho^{\gamma}
    + c^{\alpha}_{bc}\eta^b\wedge\eta^c)Y_{\alpha}
    - \sum_{\alpha<\beta}\rho^{\alpha}\wedge\rho^{\beta}c^{\gamma}_{\alpha\beta}Y_{\gamma}\\
 &=
 \sum_{b<c}c^{\alpha}_{bc}\eta^b\wedge\eta^c Y_{\alpha}
\end{align*}
which follows from formulas (\ref{e:drho}).
\end{proof}

The theorem thus provides a particular choice of a
truncating two-form. When $n=3$ this is the only possible
choice.
Indeed, this is so because a two-form in three dimensions is already fixed by
specifying its contractions (to one-forms) with respect to
two transversal vector fields. The
two vector fields are $\Xnh$ and the infinitesimal generator
of the $H$-action. Of course, one is really only interested
in the point-wise tangent projections of these vector fields.
Indeed, to tie this to \cite{N08} notice that the two-form
$-i(\Xnh)\nu$ defined in \cite{N08} is just
$\vv<J,K>-\vv<L,\curv^{\om}>$; the form
$\nu=\rho^1\wedge\eta^1\wedge\eta^2$
is the standard volume form on $S=\SO(3)$.

In higher dimensions, however, there will be many
different possibilities, and it is not clear whether these
are all on an equal footing. For example, are there choices
which yield a form $\wt{\Om}$ which becomes (conformally)
closed after restriction to a level set of $J_H$ while this
is not true for other choices?

The existence of $\wt{\Om}$ in the above proposition allows to
replace the triple\break $(TS,\Omnh,\Ham)$ with the triple
$(TS,\wt{\Om},\Ham)$. This leaves the dynamics unaltered but has the
advantage that the conserved quantity $J_H$ is now the \momap
associated to the $H$-symmetry. We can thus do (almost) Hamiltonian
reduction and pass to the quotient $J_H^{-1}(\orb)/H$ where
$\orb\subset\ho^*$ is a coadjoint orbit.

\begin{corollary}[The ultimate reduced phase space]\label{cor:4.2}
Let $\orb\subset\ho^*$ be a coadjoint orbit. Then
\[
 J_H^{-1}(\orb)/H
 \cong
 TS^{n-1}\times_{S^{n-1}}(S\times_H\orb)
\]
where the isomorphism depends on the
mechanical connection on $S\toto S/H$ associated to the
metric $\mu_0$. In particular, $J_H^{-1}(\orb)/H$ is isomorphic to a
bundle over $TS^{n-1}$ with fiber $\orb$.
\end{corollary}

\begin{proof}
This follows from the usual argument involving the mechanical connection
and the locked inertia tensor associated to $\mu_0$.
\end{proof}

Let $\lam\in\orb$.
Since $H\times_{H_{\lam}} J_H^{-1}(\lam) \cong J^{-1}_H(\orb)$
where $H_{\lam}$ is the stabilizer subgroup at $\lam$
we can also do point reduction to arrive at the same reduced space,
that is, $J_H^{-1}(\lam)/H_{\lam} = J_{H}^{-1}(\orb)/H$.
This implies the following corollary.

\begin{corollary}\label{cor:hom}
When
$\mathbb{I}=1$ Chaplygin's ball is Hamiltonian after
reduction by $H$.
\end{corollary}

\begin{proof}
In this case $L=J_H$ and closedness follows
from the Bianchi identity for the curvature form.
\end{proof}

We stress that truncation is necessary even in the
homogeneous case.
This is due to the fact that $\D$ is never the horizontal space of the
mechanical connection associated to $\mu$.
Once the non-holonomic two-from $\Omnh$ has been
altered one can perform reduction and it is only then that
the system becomes Hamiltonian.  This should be compared
with \cite[Section~3.3]{EhlersKoiller}. See also the remarks in Section~\ref{S:concl}.

\subsection{Hamiltonization of the $3$-dimensional ball}
Let
$n=3$. Consider the metric isomorphism $\Phi{} :=
(\mu_0)^{\check{}} = \mathbb{I}+\A^*\A: TS\to T^*S =_{\langle .,. \rangle} TS$,
$(s,u)\mapsto \mathbb{I}u + \sum
g_a(s,u)\Ad(s^{-1})Z_a$.
Define
\[
 f(s) = (\det \Phi{}_s)^{-\frac{1}{2}} \textup{ where } s\in S.
\]
Because of $H$-invariance $f$ drops to a function
$S^2\to\R$.\footnote{This function was called $\rho_{\mu}$ in
  \cite[Section~3]{BorisovMamaev2005} and has also been considered in
  \cite{FK95} in the context of higher dimensional Chaplygin systems.}

\begin{proposition}[Hamiltonization]\label{prop:ham}
Let
$\lam\in\ho^*\cong\R$. Then $d(f\wt{\Om})|_{J_H^{-1}(\lam)}=0$.
\end{proposition}

\begin{proof} Let $\iota: J_H^{-1}(\lam)\hookto TS$ be the
inclusion.  Notice that
\begin{align*}
  &\iota^*d(f\wt{\Om})
     = \iota^*(df\wedge\Om^S - df\wedge\vv<L,\curv^{\om}>
       - fd\vv<L,\curv^{\om}>) = 0\\
  &\iff
    df\wedge\theta^S - f\vv<L,\curv^{\om}>
  \textup{ is closed on } J_H^{-1}(\lam).
\end{align*}
Since $\mu_0$ is $H$-invariant it follows that
$\xi_{\alpha}.\Phi{} = 0$ and using the
derivation property of the determinant function we find that
\[
 df
 = -\by{1}{2}(\det \Phi{})^{-\by{3}{2}}\det(\Phi{})
      \sum\tr(\Phi{}^{-1}\zeta_a.\Phi{}) \eta^a
 = -\by{1}{2}f\sum \tr(\Phi{}^{-1}\zeta_a.\Phi{}) \eta^a.
\]
Computing the trace with respect to the orthonormal basis
$\Ad(s^{-1})Y_{\alpha},\Ad(s^{-1})Z_a$
gives
\[
 N_a
 := \tr(\Phi{}^{-1}\zeta_a.\Phi{})
  = -2\sum\vv<\Phi{}^{-1}c^{\alpha}_{ab}\Ad(s^{-1})Y_{\alpha},\Ad(s^{-1})Z_b>.
\]
Actually $\alpha=1$ and $a = 1,2$ because $n=3$. However,
for notational reasons we will make the convention that
$\alpha = 3$. The basis of $\so(3)$ is thus called
$Z_1,Z_2,Y_3$.  Therefore,
\begin{align*}
 &df\wedge\theta^S - f\vv<L,\curv^{\om}>\\
 =&
  f\big(-\by{1}{2}\sum N_a\eta^a\wedge(\wt{l}_3\rho^3
        + (\wt{g}_b+g_b)\eta^b)
    - l_3\eta^1\wedge\eta^2\big)\\
 =& -f\big(\by{1}{2}\sum N_a\wt{l}_3\eta^a\wedge\rho^3
    + (\by{1}{2}(N_1(\wt{g}_2+g_2)-N_2(\wt{g}_1+g_1))
        + l_3)\eta^1\wedge\eta^2\big)
\end{align*}
Notice that that the first term in this expression, $-\by{1}{2}f\sum
N_a\eta^a\wedge\wt{l}_3\rho^3 = df\wedge\wt{l}_3\rho^3$,
becomes closed upon
restriction to a level set of $J_H = \wt{l}_3 Y_3$.
For the middle term, a short calculation using that $n=3$ now shows that
\[
 N_1(\wt{g}_2+g_2)-N_2(\wt{g}_1+g_1)
 =
 -2l_3 - 2\vv<\Phi{}^{-1}\Ad(s^{-1})Y_3,\Ad(s^{-1})Y_3>\wt{l}_3.
\]
Therefore,
\[
 f( \by{1}{2}(N_1(\wt{g}_2+g_2) - N_2(\wt{g}_1+g_1))
   +
   l_3 )\eta^1\wedge\eta^2
 =
  -f\vv<\Phi^{-1}\Ad(s^{-1})Y_3,\Ad(s^{-1})Y_3>\wt{l}_3\eta^1\wedge\eta^2
\]
which is also closed when restricted to a level set of $J_H
= \wt{l}_3 Y_3$.
\end{proof}

This approach gives a symplecto-geometric explanation
of the formulas in \cite{BorisovMamaev,BorisovMamaev2005}.
Note in particular that the proof involves rather little computation.

Unfortunately the above proof relies very heavily on the
fact that $n=3$. However, it is designed so that, in principle,
all the
expressions also make sense in higher dimensions. It is
hoped that this approach can also be useful in studying
cases of Hamiltonization in dimensions $n>3$.
Indeed, it would be very nice if these techniques could be used to
give a useful characterization of those inertia
matrices $\mathbb{I}$ and values of $J_H$ which yield a
system that is Hamiltonizable after reduction by $H$.

\section{Comments and conclusions}\label{S:concl}

One of the goals of this paper was to work out the reduction of
$G$-Chaplygin systems with respect to additional internal
symmetries modeled by a Lie group $H$ subject to the compatibility
conditions
described in Section~\ref{S:InternalSym}.
The first step was to describe the compression to an almost
Hamiltonian system $(T^*S,\Omnh,\Hamc)$ in the presence of internal
symmetries. A construction that is similar to this step can also be
found in \cite{S02,S98}. The novelty in the
truncation procedure is that we can reduce the dynamics of the system
to a coadjoint bundle over $T^*B=T^*(S/H)$ and \emph{reproduce} the
structure of an almost Hamiltonian system. This gives a general answer
to a question posed for the special case of the $3$-dimensional
Chaplygin ball problem in \cite[Section~4.1]{EKR03}.

The main technical step in our reduction procedure is called
\emph{truncation}.
This involves a choice of a principal bundle connection $\sigma$ on
$T^*S\toto(T^*S)/H$ such that the non-holonomic vector field $\Xnh$ is
horizontal.
The name is chosen because, effectively, we use
the connection $\sigma$ to cut off all the information contained
in the $\vv<J,K>$-term that is not seen by the dynamics but presents
an obstruction to reduction.

In Section~\ref{S:Chaplygin} we apply this reduction procedure to the
$n$-dimensional Chaplygin ball problem. Thus
we write the system as an almost Hamiltonian system on a coadjoint
bundle over $T^*(S^{n-1})$.
In particular
we derive a symplectic
proof of the
remarkable result of \cite{BorisovMamaev,BorisovMamaev2005} on the
Hamiltonizability of the $3$-dimensional Chaplygin ball.

Furthermore, we can also deal with the $n$-dimensional homogeneous
Chaplygin ball. In Corollary~\ref{cor:hom} we show that this system
is Hamiltonian after reduction of internal symmetries (but not at the
compressed level).
From the mathematical point of view this is a non-trivial conclusion:
even in the homogeneous case the
connection $\D$ does not coincide with the mechanical connection
associated to $\mu$, whence one cannot employ usual symplectic
reduction techniques to construct the reduced phase space.
In fact, \cite{EhlersKoiller}
have shown (for $n=3$) that the problem is not even Hamiltonizable
(i.e., conformally symplectic) at the compressed
level. Thus one has to use truncation to eliminate the internal
symmetries, and it is only then that the system becomes Hamiltonian.
On the other hand, the result is obvious from a physical perspective:
Consider the big phase space $T^*Q=T^*(S\times\R^{n-1})$ and the
Hamiltonian $\Ham$ of the ball. Let $X_{\mathcal{H}}$ denote the
Hamiltonian vector field associated to $\Ham$ with respect to the
\emph{canonical symplectic} structure on $T^*Q$. This is the
homogeneous $n$-dimensional ball that rolls on a horizontal table
without constraints. If this ball happens to satisfy the no-slip
condition at \emph{one} time instant it will also have to satisfy
the constraints for all future and past time; it cannot accelerate and
will roll on a straight line. The point is that this physical
fact cannot be described in the framework of existing reduction theories:
either one does symplectic reduction of the free system but then one
cannot describe the constraints, i.e., the space $\D$ (which could be
viewed as necessary initial conditions), within this process; or one
does compression which captures the constraint space $\D$ but destroys
the Hamiltonian feature of the system. Hence the need for
truncation.
See Corollary~\ref{cor:hom}.

The  truncation of  $\Omnh$ is an example of a more general procedure
in which one consistently replaces the
almost Hamiltonian system $(T^*S,\Omnh,\Hamc)$  by
$(T^*S,\wt{\Om},\Hamc)$.
Even though
both systems define the same vector field on $T^*S$, there may
an advantage in working with $\wt{\Om}$. (For instance, one may be conformally
symplectic while the other is not.)
This is the idea of adding an \emph{affine term} to $\Omnh$ which
seems to go back to \cite{St85}, has been formalized in
\cite{EhlersKoiller}, and successfully used in \cite{N08}.
An affine term is a semi-basic two form on $T^*S$ which vanishes when
contracted with $\Xnh$.
The problem
is how to choose the affine term. In the special case of internal
symmetries the situation is easier as the symmetries provide extra
information.
Notice that in Section~\ref{S:trunc-ball} we used the truncation to
find our choice of affine term. In Theorem~\ref{prop:trunc}, however,
we did not
use the truncated two-form $\vv<J,K>\circ\Lam^2\chi$,
but rather another form that we found to be more convenient.
Thus it is
important to remember that one has many different possibilities here
and the truncation is just a means to find one particular choice.
More generally,
the idea of modifying $\Omnh$ seems to be
important also for systems without internal symmetries
(such as the rubber ball) but a
systematic treatment
is not known.
The Dirac reduction techniques (which do not use internal symmetries)
developed in \cite{JR08} could provide a  starting point, but it seems to us
that one encounters  the same difficulties as in compression.

The study of
other non-holonomic systems, including the rubber ball, with this
perspective is work in progress.

\textbf{Thanks.}
We would like to thank the referees for their detailed reports and
many constructive suggestions.

\medskip
Received November 2008; revised March 2009.
\medskip

\end{document}